\documentclass[final]{amsart}
\usepackage{amsmath,amssymb,amsfonts,amsthm,hyperref}

\usepackage{lmodern}
\usepackage[T1]{fontenc}

\usepackage{tikz-cd}
\usetikzlibrary{matrix,arrows,decorations.pathmorphing}

\theoremstyle{plain}
\newtheorem{theorem}{Theorem}[section]
\newtheorem{proposition}[theorem]{Proposition}

\theoremstyle{remark}
\newtheorem{remark}{Remark}
\newtheorem{example}{Example}

\newcommand{\RR}{\mathbb R}

\newcommand{\HH}{\mathcal{H}}
\newcommand{\KK}{\mathcal{K}}
\newcommand{\D}{\mathcal{D}}
\renewcommand{\S}{\mathcal{S}}
\renewcommand{\P}{\mathcal{P}}
\newcommand{\A}{\mathcal{A}}
\renewcommand{\L}{\mathcal{L}}
\newcommand{\E}{\mathcal{E}}
\newcommand{\U}{\mathcal{U}}
\newcommand{\G}{\mathcal{G}}
\renewcommand{\u}{\mathfrak{u}}
\newcommand{\g}{\mathfrak{g}}
\newcommand{\T}{\operatorname{T}\!}
\renewcommand{\H}{\operatorname{H}\!}
\newcommand{\V}{\operatorname{V}\!}
\newcommand{\Ker}{\operatorname{Ker}}
\newcommand{\Ad}{\operatorname{Ad}}
\newcommand{\Tr}{\operatorname{Tr}}
\newcommand{\eps}{\varepsilon}
\newcommand{\1}{{\mathbf 1}}
\newcommand{\0}{{\mathbf 0}}
\newcommand{\dd}[1]{\frac{\operatorname{d}}{\operatorname{d}\!#1}}
\newcommand{\ket}[1]{\lvert #1 \rangle}
\newcommand{\bra}[1]{\langle #1 \rvert}

\newcommand{\ketbra}[2]{\lvert #1 \rangle \langle #2 \rvert}
\newcommand{\obs}[1]{\mathbf{#1}}

\begin{document}

\title[Geometric uncertainty relation for mixed quantum states]{Geometric uncertainty relation for mixed quantum states}

\author{Ole Andersson}
\email{olehandersson@gmail.com}
\author{Hoshang Heydari}
\thanks{The second author acknowledges the financial support from the Swedish Research Council (VR), grant number
2008-5227.}
\email{hoshang@fysik.su.se}
\address{Department of Physics, Stockholm University, 10691 Stockholm, Sweden}

\date{\today}

\begin{abstract}
In this paper we use symplectic reduction in an Uhlmann bundle to construct a principal fiber bundle over a general space of unitarily equivalent mixed quantum states. The bundle, which generalizes the Hopf bundle for pure states, gives in a canonical way rise to a Riemannian metric and a symplectic structure on the base space. With these we derive a geometric uncertainty relation for observables acting on quantum systems in mixed states. We also give a geometric proof of the classical Robertson-Schr{\"o}dinger uncertainty relation, and we compare the two. 
They turn out not to be equivalent, because of the multiple dimensions of the gauge group for general mixed states.
We give examples of observables for which the geometric relation provides a stronger estimate than that of Robertson and Schr\"{o}dinger, and vice versa.
\end{abstract}
\maketitle 

\section{Introduction}\label{introduction}
Modern methods in quantum mechanics and quantum information theory have given rise to several generalizations of the classic Heisenberg-Kennard-Weyl \cite{Heisenberg1927, Kennard1927, Weyl1928} and Robertson-Schr{\"o}dinger \cite{Robertson_1929, Schrodinger_1930} uncertainty relations.
These include relations involving the Wigner-Yanase skew-information \cite{Lou2003,Lou2005, Park2005}, purity-bounded  \cite{Dodonov2002} and entropic uncertainty relations \cite{Wehner_etal2010}, and uncertainty relations for non-Hamiltonian systems \cite{Ingarden1973,Dodonov_etal1980,Tarasov2013}.
In this paper,  we adopt a symplectic geometric approach  to uncertainty relations for quantum systems in mixed states.

Common for classical and pure quantum systems is that their phase spaces are symplectic manifolds and their observables generate symplectic flows \cite{Gunter_1977,Kibble_1979,Ashtekar_etal1998,Brody_etal1999}.
Quantum systems, however, exhibit characteristics that have no classical counterparts.
One is the impossibility to fully predict results of measurements.

Observables in classical mechanics are represented by real-valued functions, and the result of a measurement of an 
observable equals the value of the corresponding function at the point in phase space that labels the system's state. This means that measurement results are completely predictable in classical mechanics.
The situation in quantum mechanics is quite different. There, observables are represented by self-adjoint operators, 
and only expectation values and uncertainties of observables can be calculated. The actual value of an 
observable cannot in general be known prior to measurement.
Furthermore, there is a limit to the precision with which values of pairs of observables can be known simultaneously.
This is the famous quantum uncertainty principle. 

A quantum system prepared in a pure state can be modeled on a projective Hilbert space equipped with the Fubini-Study K\"{a}hler metric. The 
real and imaginary parts of this metric (suitably scaled) equips the projective Hilbert space with Riemannian and symplectic structures, and thus with Riemann and Poisson brackets. 
For observables acting on systems in pure states it has been shown \cite{Ashtekar_etal1998} that the Robertson-Schr\"{o}dinger uncertainty relation can be expressed entirely in terms of the Riemann and Poisson brackets of the observables' expectation value functions. In this paper we generalize this result to systems in mixed states.

The state of an experimentally prepared quantum system usually exhibits classical uncertainty -- it is \emph{mixed}, 
i.e., an incoherent superposition of pure states. Mixed states are commonly represented by density operators.
The spectrum of a density operator representing a quantum system's state is preserved when the system evolves according to a von Neumann equation. In this paper we equip the spaces of isospectral finite rank density operators with natural Riemannian and symplectic structures and we derive a geometric uncertainty relation for observables acting on systems in mixed states\footnote{By a mixed state we henceforth mean a \emph{finite} incoherent superposition of pure states. Such a mixed state can be represented by a density operator of finite rank.}.
We also give a geometric proof of the Robertson-Schr\"{o}dinger uncertainty relation, and we compare the two uncertainty relations.
It turns out that in general they are not equivalent. We show explicitly why this is so and we give examples of observables for which the geometric relation provides a stronger estimate than that of Robertson and Schr\"{o}dinger, and vice versa.
This work is inspired by \cite{Montgomery1991}.

\section{Purification of mixed states and symplectic reduction}
This paper deals with quantum systems in mixed states which evolves unitarily.
The systems will be modeled on an unspecified Hilbert space $\HH$ and their states will be represented by density operators.
Recall that a density operator is a positive trace-class operator with unit trace. 
We write $\D(\HH)$ for the space of density operators on $\HH$, and $\D_k(\HH)$ for the space of density operators on $\HH$ which has finite 
rank at most $k$.

A density operator whose evolution is governed by a von Neumann equation will remain in a single orbit for the left conjugation-action of the 
unitary group $\U(\HH)$ of $\HH$ on $\D(\HH)$. The orbits are in one-to-one correspondence with the possible spectra of density operators on 
$\HH$. We show in this section how symplectic reduction on the total spaces of suitably chosen Uhlmann purifications produces natural symplectic 
and Riemannian structures on the orbits of density operators of finite rank.

\subsection{Uhlmann purification of mixed states}
A quantum state is called pure if it can be represented by a single unit vector, or, equivalently, if its density operator has $1$-dimensional 
support. The term ``purification'' refers to the fact that every density operator can be regarded as a reduced pure state. 
Let $k$ be a positive integer and $\KK$ be a $k$-dimensional Hilbert space.
A particular purification reduction of mixed states represented by density operators of rank at most $k$, which we call the Uhlmann purification (for rank $k$ density operators)
\cite{Uhlmann1976,Uhlmann1986,Uhlmann1989,Uhlmann1991}, is given by the surjective map
\begin{equation}\label{Uhlmann purification}
\S(\KK,\HH)\longrightarrow\D_k(\HH),\qquad\psi\mapsto\psi\psi^\dagger.
\end{equation}
The domain of the Uhlmann purification is the unit sphere in the space $\L(\KK,\HH)$ of linear operators from $\KK$ to $\HH$ equipped with the 
Hilbert-Schmidt Hermitian inner product.

\begin{remark}
If a quantum system modeled on $\HH$ can be considered a subsystem of a larger quantum system, 
the Hilbert space of the larger system decomposes as a product $\HH\otimes\KK$ and the subsystem's 
state is represented by the partial trace with respect to $\KK$ of the density operator representing the larger system's state. 
Partial trace reduction is justified by the fact that it is the only reduction of density 
operators on $\HH\otimes\KK$ that respects expectation values of local observables \cite{Nielsen_etal2010}: If $\obs{A}$ is an observable on 
$\HH$ and 
$\rho=\Tr_{\KK}R$, where $R$ is a density operator on $\HH\otimes\KK$, then $\Tr(\obs{A}\rho)=\Tr((\obs{A}\otimes\obs{1})R)$. 
Now, the Uhlmann purification is a special case of partial trace reduction. If we identify $\L(\KK,\HH)$ with $\HH\otimes\KK^*$, write 
$\S(\HH\otimes\KK^*)$ for the unit sphere in $\HH\otimes\KK^*$, and $\P(\HH\otimes\KK^*)$ for the projective space over $\HH\otimes\KK^*$, 
then \eqref{Uhlmann purification} is given by the composition 
\begin{equation} \label{pur}
\begin{tikzcd}
\S(\HH\otimes\KK^*) \arrow{rr}{\ket{\psi}\mapsto\ketbra{\psi}{\psi}} & &\P(\HH\otimes\KK^*) \arrow{r}{\Tr_{\KK^*}}  &\D_k(\HH).
\end{tikzcd}
\end{equation}
\end{remark}

\subsection{Reduced bundles of purifications}
By the spectrum of a density operator with $k$-dimensional support we mean the sequence 
$\sigma=(p_1, p_2, \dots, p_k)$ 
of the density operator's \emph{positive} eigenvalues listed in descending order. Throughout the 
rest of this paper we fix such a spectrum $\sigma$ and write $\D(\sigma)$ for the corresponding $\U(\HH)$-orbit in $\D_k(\HH)$. Furthermore, we 
fix a $k$-dimensional Hilbert space $\KK$ and an orthonormal basis $\{\ket{1},\ket{2},\dots,\ket{k}\}$ in $\KK$.

Let $\S(\sigma)$ be the set of maps $\psi$ in $\L(\KK,\HH)$ satisfying $\psi^\dagger\psi=P$, where 
\begin{equation}
P=\sum_{j=1}^kp_j\ketbra{j}{j},
\end{equation}
and define 
$\pi:\S(\sigma)\to\D(\sigma)$ as the restriction of the Uhlmann purification to $\S(\sigma)$.
Then $\pi$ is a principal fiber bundle with gauge group
$\U(\sigma)$ consisting of all unitary operators on $\KK$ which commutes with $P$. The action of $\U(\sigma)$ on $\S(\sigma)$ is induced by the 
canonical right action of the unitary group $\U(\KK)$ on $\L(\KK,\HH)$. We write $\u(\sigma)$ for the Lie algebra of $\U(\sigma)$. This algebra consists of all anti-Hermitian operators on $\KK$ which commutes with $P$.

\begin{example}
If the operators in $\D(\sigma)$ represents pure states, i.e., if $\sigma=(1)$, then $\S(\sigma)$ is the unit sphere in $\HH$, $\D(\sigma)$ is the projective space over $\HH$, and $\pi$ is the generalized Hopf bundle.
\end{example}

\begin{remark}
For a general spectrum the total space $\S(\sigma)$ is diffeomorphic to the Stiefel variety of $k$-frames in $\HH$. But $\S(\sigma)$ will be 
equipped with a Riemannian metric which is different from the metric on the Stiefel variety induced by the standard bi-invariant metric on 
$\U(\HH)$.
\end{remark}

When you purify a state you artificially add extraneous, perhaps inaccessible, information to your state. One way to get rid of this additional 
information is to use symplectic reduction. Next we will show that $\pi$ is a reduced phase space submersion generated by a 
symplectic reduction on $\L(\KK,\HH)$ by the action of $\U(\KK)$. The existence of a symplectic structure on $\D(\sigma)$ will then follow from 
the classic symplectic reduction theorem by Marsden and Weinstein \cite{Marsden_etal1974}:

\begin{theorem}\label{symplectic reduction theorem}
Consider a symplectic manifold $(M,\Omega)$ on which there is a Hamiltonian right action of a Lie group $\G$ 
with a coadjoint-equivariant momentum map $J: M\to\g^*$. Assume that $\mu$ is a regular value of $J$ 
and that the isotropy group $\G_\mu$ of $\mu$ acts freely and properly on $J^{-1}(\mu)$. Then the orbit projection $\pi_\mu:J^{-1}(\mu)\to 
J^{-1}(\mu)/\G_\mu$ is a principal $\G_\mu$-bundle, and 
the reduced phase space $J^{-1}(\mu)/\G_\mu$ has a unique symplectic form $\omega_\mu$ characterized by 
$\pi_\mu^*\omega_\mu=\Omega|_{J^{-1}(\mu)}$.
\end{theorem}
Recall that a right action by a Lie group $\G$ on a symplectic manifold $(M,\Omega)$ is Hamiltonian if the action is by symplectic 
transformations $R_g$ and it has a coadjoint-equivariant momentum map $J:M\to\g^*$, thus $J\circ R_g=\Ad_g^*\circ J$ where $g\mapsto 
\Ad_g^*$ is the right coadjoint representation of $\G$ on $\g^*$. 
Recall also the following effective way to determine if $\mu$ is a regular value of $J$: The symmetry algebra of $x$ in $M$ is the subalgebra 
$\g_x$ of $\g$ consisting of those elements whose corresponding fundamental vector fields vanishes at $x$. Now $\mu$ is a regular value for $J$ 
if the symmetry algebra  for every $x$ in $J^{-1}(\mu)$ vanishes:

\begin{proposition}
A functional $\mu$ is a regular value of $J$ if $\g_x=\{0\}$ for every $x$ in $J^{-1}(\mu)$.
\end{proposition}

\begin{proof}
A point $x$ in $M$ is regular if $dJ$ maps the tangent space at $x$ onto $\g^*$.
This is equivalent to $\0$ being the only element $\xi$ in $\g$ that satisfies $dJ(X)(\xi)=0$ for all tangent vectors $X$ at $x$.
Now $dJ(X)(\xi)=\Omega(\hat{\xi}(x),X)$, where $\hat{\xi}$ is the fundamental vector field on $M$ corresponding to $\xi$. The proposition thus 
follows from the nondegeneracy of $\Omega$.
\end{proof}

\subsubsection{Symplectic structure}
A Riemannian metric and a symplectic form on $\L(\KK,\HH)$ is given by $2\hbar$ times the real part and the imaginary part, respectively, of the 
Hilbert-Schmidt inner product:
\begin{equation}\label{Gustaf Otto}
G(X,Y)=\hbar\Tr(X^\dagger Y + Y^\dagger X),\qquad \Omega(X,Y)=-i\hbar\Tr(X^\dagger Y-Y^\dagger X).
\end{equation}
Furthermore, the unitary group $\U(\KK)$ acts on $\L(\KK,\HH)$ from the right by isometric and symplectic transformations, $R_U(\psi)= \psi U$.
We write $\u(\KK)$ for the space of anti-Hermitian operators on $\KK$, i.e., the Lie algebra of $\U(\KK)$, and $\hat{\xi}$ for the fundamental 
vector field corresponding to $\xi$ in $\u(\KK)$:
\begin{equation}
\hat{\xi}(\psi)=\dd{t}\Big[R_{\exp(t\xi)}(\psi)\Big]_{t=0}=\psi\xi.
\end{equation}
Moreover, we write $\u(\KK)^*$ for the space of functionals on $\u(\KK)$ and recall that every member in $\u(\KK)^*$ 
is of the form $\mu_\eta$ for some Hermitian operator $\eta$ on $\KK$, where 
$\mu_\eta(\xi)=i\hbar\Tr(\eta\xi)$.
Define $J:\L(\KK,\HH) \to \u(\KK)^*$ by $J(\psi)=\mu_{\psi^\dagger\psi}$.
\begin{proposition}
$J$ is a coadjoint-equivariant momentum map for the $\U(\KK)$ action on $\L(\KK,\HH)$.
\end{proposition}

\begin{proof}
Let $\xi$ be an anti-Hermitian operator on $\KK$, $X$ be a vector tangent to $\L(\KK,\HH)$ at $\psi$, and 
$U$ be a unitary operator on $\KK$.
Define $J_\xi:\L(\KK,\HH)\to\RR$ by $J_\xi(\psi)=J(\psi)(\xi)$. Then,
\begin{align}
&dJ_\xi(X)=i\hbar\Tr(X^\dagger\psi\xi+\psi^\dagger X\xi)=\Omega(\hat{\xi}(\psi),X),\\
&J(R_U(\psi))(\xi)=i\hbar \Tr((\psi U)^\dagger (\psi U)\xi)=i\hbar\Tr(\psi^\dagger\psi U\xi U^\dagger)=\Ad^*_U J(\psi)(\xi).
\end{align}
Thus, $\hat{\xi}$ is the Hamiltonian vector field for $J_\xi$, and $J$ is coadjoint-equivariant.
\end{proof}

Let $\E(\KK)$ be the space of functionals on $\u(\KK)$ that are represented by density operators on $\KK$.
We have that $J$ maps $\S(\KK,\HH)$ into $\E(\KK)$ because $\psi^\dagger\psi$ is a density operator on $\KK$ if $\psi$ is a normalized linear 
map from $\KK$ into $\HH$. In fact, $J$ maps $\S(\KK,\HH)$ \emph{onto} $\E(\KK)$, and $J^{-1}(\mu_\eta)$ is a subset of $\S(\KK,\HH)$ for every 
density operator $\eta$ on $\KK$ because $\mu_\eta=\mu_{\psi^\dagger\psi}$ implies $\Tr(\psi^\dagger\psi)=\Tr(\eta)=1$.
Two density operators on $\KK$ are unitarily equivalent, or, equivalently, their corresponding functionals belong to the same coadjoint orbit in 
$\u(\KK)^*$, if and only if they have identical spectra. The orbit corresponding to the spectrum $\sigma$ contains the density operator $P$ and 
$\S(\sigma)=J^{-1}(\mu_P)$.

\begin{proposition}
The functional $\mu_P$ is a regular value for $J$, the isotropy group of $\mu_P$ is $\U(\sigma)$, and $\U(\sigma)$ acts freely and properly on 
$\S(\sigma)$.
\end{proposition}

\begin{proof}
Suppose $J(\psi)=\mu_P$ and assume $\xi$ in $\u(\KK)$ is such that $\hat{\xi}(\psi)=\0$. Then $\xi=P^{-1}\psi^\dagger\psi\xi=\0$. Thus the 
symmetry algebra of $\psi$ is trivial.
Moreover, the isotropy group of $\mu_P$ is $\U(\sigma)$ because $i\hbar\Tr(PU\xi U^\dagger)=i\hbar\Tr(P\xi)$ for every $\xi$ in $\u(\KK)$ is equivalent to the assertion that $U$ commutes with $P$. Finally, $\U(\sigma)$ acts freely on $\S(\sigma)$ since $\psi 
U=\psi$ implies $U=P^{-1}\psi^\dagger\psi U=P^{-1}\psi^\dagger\psi=\1$, and properly because it is compact.
\end{proof}

\begin{theorem}\label{symplectic main theorem}
$\D(\sigma)$ admits a unique symplectic form $\omega$ such that $\pi^*\omega=\Omega|_{\S(\sigma)}$.
\end{theorem}

\begin{proof}
The orbit projection $\pi_{P}:\S(\sigma)\to\S(\sigma)/\U(\sigma)$ is a principal $\U(\sigma)$-bundle, and the reduced phase space 
$\S(\sigma)/\U(\sigma)$ has a unique symplectic form $\omega_{P}$ such that $\pi_{P}^*\omega_{P}=\Omega|_{\S(\sigma)}$ by theorem 
\ref{symplectic reduction theorem}.  
The projection $\pi$, which maps $\S(\sigma)$ onto $\D(\sigma)$, factorizes through a surjection $\S(\sigma)/\U(\sigma)\to\D(\sigma)$.
This map is also injective. 
For if $\psi\psi^\dagger=\phi\phi^\dagger$, then $\phi=\psi U$ where $U=\psi^\dagger\phi P^{-1}$. The operator $U$ is unitary
and commutes with $P$:
\begin{align}
&U^\dagger U=P^{-1}\phi^\dagger\psi\psi^\dagger\phi P^{-1}=P^{-1}\phi^\dagger\phi\phi^\dagger\phi P^{-1}=P^{-1} P P P^{-1}=\1,\\
&P U=P\psi^\dagger\phi P^{-1}=\psi^\dagger\psi \psi^\dagger\phi P^{-1}=\psi^\dagger\phi \phi^\dagger\phi P^{-1}=\psi^\dagger\phi=UP.
\end{align} 
So $U$ belongs to $\U(\sigma)$.
\end{proof}

\begin{remark}
Another choice of representative than $P$ for the orbit corresponding to $\sigma$ gives  rise to an isomorphic bundle. For if $\tilde P$ is 
another density operator on $\KK$ with spectrum $\sigma$, then $\tilde P=U P U^\dagger$ for some unitary operator on $\KK$. Set  
$\tilde{\S}(\sigma)=\{\psi\in\L(\KK,\HH):\psi^\dagger\psi=\tilde{P}\}$ and define $\tilde{\pi}:\tilde{\S}(\sigma)\to\D(\sigma)$ by 
$\tilde{\pi}(\psi)=\psi\psi^\dagger$. The map $\tilde{\pi}$ is well defined because $\psi\psi^\dagger$ and $\tilde{P}$ have the same spectrum, 
namely $\sigma$, and we have a commutative diagram:
\begin{center}
\begin{tikzcd}
\tilde{\S}(\sigma) \arrow{dr}[swap]{\tilde{\pi}} \arrow{rr}{R_U} & {} & \S(\sigma) \arrow{dl}{\pi}\\
{} & \D(\sigma) & {}
\end{tikzcd}
\end{center}
The gauge groups of $\tilde{\pi}$ and $\pi$ are conjugate equivalent,  $\tilde{\U}(\sigma)=U\U(\sigma)U^\dagger$,
and $R_U$ is an isometric bundle map with respect to $G|_{\tilde{\S}(\sigma)}$ and $G|_{\S(\sigma)}$ such that $R_U^*\Omega|_{\S(\sigma)}=\Omega|_{\tilde{\S}(\sigma)}$.
\end{remark}

\subsubsection{Riemannian structure}
The metric $G$ restricts to a gauge-invariant metric on $\S(\sigma)$. We define the vertical and horizontal bundles over $\S(\sigma)$ to be the 
subbundles $\V\S(\sigma)=\Ker d\pi$ and $\H\S(\sigma)=\V\S(\sigma)^\bot$
of the tangent bundle $\T\S(\sigma)$. Here $d\pi $ is the differential of $\pi$ and $^\bot$ denotes the orthogonal complement with respect to 
$G$. Vectors in $\V\S(\sigma)$ and $\H\S(\sigma)$
are called vertical and horizontal, respectively. We equip $\D(\sigma)$ with the unique metric $g$ which makes $\pi$ a Riemannian submersion. Thus, $g$ is such that the restriction of $d\pi$ to every fiber of $\H\S(\sigma)$ an isometry.

The fundamental vector fields of the gauge group action on $\S(\sigma)$ yield canonical isomorphisms between $\u(\sigma)$ and the fibers in 
$\V\S(\sigma)$. Furthermore, $\H\S(\sigma)$ is the kernel bundle of the gauge invariant mechanical connection $\A:\T\S(\sigma)\to\u(\sigma)$ 
defined by 
$\A_{\psi}=\mathbb{I}_{\psi}^{-1}\mathbb{J}_{\psi}$,
where $\mathbb{I}:\S(\sigma)\times\u(\sigma)\to \u(\sigma)^*$ and $\mathbb{J}:\T\S(\sigma)\to \u(\sigma)^*$ are the locked inertia tensor and 
metric momentum map, respectively:
\begin{equation}
\mathbb{I}_{\psi}\xi(\eta)=G(\hat\xi(\psi),\hat{\eta}(\psi)),\qquad 
\mathbb{J}_{\psi}(X)(\xi)=G(X,\hat\xi(\psi)).
\end{equation}
The inertia tensor is of constant bi-invariant type since $\mathbb{I}_{\psi}$ is an adjoint-invariant form on $\u(\sigma)$ which is independent 
of $\psi$. Thus it defines a metric on $\u(\sigma)$:
\begin{equation}\label{beta}
\xi\cdot\eta=\hbar\Tr\Big(\big(\xi^\dagger \eta+\eta^\dagger \xi\big)P\Big).
\end{equation}
We use this metric to derive an explicit formula for the mechanical connection.
The formula involves the operators
\begin{equation}
\Pi_j=\sum_{i=m_1+\dots+m_{j-1}+1}^{m_1+\dots+m_j}\ketbra{i}{i}\qquad (j=1,2,\dots,l),
\end{equation}
where $m_1, m_2, \dots , m_l$ are the multiplicities of the different eigenvalues in $\sigma$, with $m_1$ being the multiplicity of the greatest eigenvalue, $m_2$ the multiplicity of the second greatest eigenvalue, etc.
\begin{proposition}\label{connect}
Let $X$ be a tangent vector to $\S(\sigma)$ at $\psi$. Then
\begin{equation}
\A_\psi(X)=\sum_{j=1}^l\Pi_j\psi^\dagger X\Pi_jP^{-1}.
\end{equation}
\end{proposition}
\begin{proof}
We note first that each $\Pi_j\psi^\dagger X\Pi_jP^{-1}$ belongs to $\u(\sigma)$. For
\begin{equation}\label{eq}
\Pi_j\psi^\dagger X\Pi_jP^{-1}=p_j^{-1}\Pi_j\psi^\dagger X\Pi_j=P^{-1}\Pi_j\psi^\dagger X\Pi_j
\end{equation}
implies that $\Pi_j\psi^\dagger X\Pi_jP^{-1}$ commutes with $P$. Furthermore, 
$\psi^\dagger\psi=P$ implies $X^\dagger\psi+\psi^\dagger X=\0$, which in turn implies that $\Pi_j\psi^\dagger X\Pi_jP^{-1}$ 
is anti-Hermitian:
\begin{equation}
\left(\Pi_j\psi^\dagger X\Pi_jP^{-1}\right)^\dagger+\Pi_j\psi^\dagger X\Pi_jP^{-1}=\Pi_j(X^\dagger\psi+\psi^\dagger X)\Pi_jP^{-1}=\0.
\end{equation}
The proposition now follows from
\begin{equation}
\begin{split}
\sum_{j=1}^l\Pi_j\psi^\dagger X\Pi_jP^{-1}\cdot\xi
&=\hbar\Tr\Big(\sum_{j=1}^l \Pi_j(X^\dagger\psi \xi+\xi^\dagger\psi^\dagger X)\Pi_j\Big)\\
&=\hbar\Tr\big(X^\dagger\psi\xi+\xi^\dagger\psi^\dagger X\big)\\
&=\mathbb{J}_\psi(X)(\xi)
\end{split}
\end{equation}
\end{proof}

\section{Geometric uncertainty relation}
The expectation value function of an observable $\obs{A}$ on $\HH$ is the function $A$ on $\D(\sigma)$ defined by $A(\rho)=\Tr(\obs{A}\rho)$.
We write $X_A$ for the Hamiltonian vector field of $A$. This field has a distinguished gauge-invariant lift $X_{\obs{A}}$ to $\S(\sigma)$,
$X_{\obs{A}}(\psi)=\obs{A}\psi/i\hbar$,
and we say that $\obs{A}$ is \emph{parallel} at a density operator $\rho$ if $X_{\obs{A}}$ is horizontal at some, hence every, purification in 
the fiber over $\rho$. Below we will show that the uncertainty of a parallel observable is proportional to the norm of the Hamiltonian 
vector field of the observable's expectation value function. 

The locked inertia tensor can be used as a tool to measure deviation from parallelism:
Given an observable $\obs{A}$ we define a $\u(\sigma)$-valued field $\xi_A$ on $\D(\sigma)$ by $\pi^*\xi_A=\A\circ X_{\obs{A}}$. 
Then $\xi_A\cdot\xi_A$ equals the square of the norm of the vertical part of $X_{\obs{A}}$. (Recall that $\cdot$ is the metric on $\u(\sigma)$ 
given by \eqref{beta}.) It is an interesting fact that $\xi_A$ is an intrinsic field for the quantum system that contains complete information 
about the expectation values of $\obs{A}$, c.f. \eqref{geomexp} below.
The opposite of parallelism we call \emph{perpendicularity}. Thus, $\obs{A}$ is perpendicular at $\rho$ if $X_{\obs{A}}$ is vertical along the 
fiber over $\rho$, or, equivalently, if $X_{\obs{A}}(\psi)=\psi\xi_A(\rho)$ for every lift $\psi$ of $\rho$. Note that $\obs{A}$ is 
perpendicular at $\rho$ provided that $\rho$ represents a mixture of eigenstates of $\obs{A}$.

\subsection{A geometric uncertainty relation}
The precision to which the value of an observable $\obs{A}$ can be known is quantified by its uncertainty function,
\begin{equation}
\Delta A(\rho)=\sqrt{\Tr(\obs A^2\rho)-\Tr(\obs A\rho)^2}.
\end{equation}
Furthermore, the precision to which the values of two observables $\obs{A}$ and $\obs{B}$ can be 
known simultaneously is limited by the Robertson-Schr\"{o}dinger uncertainty relation \cite{Robertson_1929,Schrodinger_1930}:
\begin{equation}\label{Sur}
\Delta A\Delta B\geq\sqrt{\left((A,B)-AB\right)^2+[A,B]^2}.
\end{equation}
Here $(A,B)$ and $[A,B]$ are the expectation value functions of the symmetric and antisymmetric products of $\obs{A}$ and $\obs{B}$:
\begin{equation}
(\obs{A},\obs{B})=\frac{1}{2}(\obs{A}\obs{B}+\obs{B}\obs{A}),\qquad
[\obs{A},\obs{B}]=\frac{1}{2i}(\obs{A}\obs{B}-\obs{B}\obs{A}).
\end{equation}
We derive a lower bound for $\Delta A\Delta B$ that involves only the Riemann and Poisson brackets of $A$ and $B$:
\begin{equation}
\{A,B\}_g=g(X_A,X_B),\qquad \{A,B\}_\omega=\omega(X_A,X_B).
\end{equation}
Thus we derive a geometric uncertainty relation for quantum systems in mixed states.
\begin{theorem}\label{geom uncert}
Let $\obs{A}$ and $\obs{B}$ be observables on $\HH$. Then 
\begin{equation}\label{uncertainty}
\Delta A\Delta B\geq \frac{\hbar}{2}\sqrt{\{A,B\}_g^2+\{A,B\}_\omega^2}.
\end{equation}
\end{theorem} 
\noindent For systems in pure states the uncertainty relation \eqref{uncertainty} agrees with the geometric uncertainty relation derived in 
\cite{Ashtekar_etal1998}. Moreover, it is a geometric version of \eqref{Sur}. But, as we will see, for general mixed states 
the two uncertainty relations are not equivalent. The difference is due to the multiple dimensions of the vertical bundle of $\S(\sigma)$.

\begin{proof}[Proof of theorem \ref{geom uncert}]
The expectation value functions of $\obs{A}$ and $\obs{B}$ are proportional to the lengths of the projections of $\xi_A$ and 
$\xi_B$, respectively, on the unit vector $\chi=\1/i\sqrt{2\hbar}$ in $\u(\sigma)$:
\begin{equation}
A=\sqrt{\frac{\hbar}{2}}\chi\cdot \xi_A, \qquad B=\sqrt{\frac{\hbar}{2}}\chi\cdot \xi_B. \label{geomexp}
\end{equation}
To see this suppose $\rho$ is a density operator in $\D(\sigma)$ and $\psi$ in $\S(\sigma)$ is such that $\rho=\psi\psi^\dagger$. By 
proposition \ref{connect} and \eqref{beta},
\begin{equation}
A(\rho)
=i\hbar\Tr(\A_\psi(X_{\obs{A}}(\psi))P)
=i\hbar\Tr(\xi_A(\rho)P)
=\sqrt{\frac{\hbar}{2}}\chi\cdot\xi_A(\rho),
\end{equation}
and similarly for $B(\rho)$.
Furthermore, 
\begin{equation}\label{geoprod}
(A,B)=\frac{\hbar}{2}\left(\{A,B\}_g+\xi_A\cdot\xi_B\right),
\qquad
[A,B]=\frac{\hbar}{2} \{A,B\}_\omega,
\end{equation}
because
\begin{align}
(A,B)(\rho)
&=\frac{\hbar}{2}G(X_\obs{A}(\psi),X_\obs{B}(\psi))
=\frac{\hbar}{2}g(X_A(\rho),X_B(\rho))+\frac{\hbar}{2}\xi_A(\rho)\cdot\xi_B(\rho),\\
[A,B](\rho)
&=\frac{\hbar}{2}\Omega(X_\obs{A}(\psi),X_\obs{B}(\psi))
=\frac{\hbar}{2}\omega(X_A(\rho),X_B(\rho)).
\end{align}
Denote the projections of $\xi_A$ and $\xi_B$ on the 
orthogonal complement of $\chi$ by $\xi_A^\bot$ and $\xi_B^\bot$. Then   
\begin{equation}\label{covariance}
(A,B)-AB
=\frac{\hbar}{2}\left(\{A,B\}_g+
\xi_A^\bot\cdot\xi_B^\bot\right)
\end{equation}
by \eqref{geomexp} and \eqref{geoprod}. In particular,
\begin{equation}
\Delta A^2=(A,A)-AA\geq\frac{\hbar}{2}\{A,A\}_g.\label{main ett}
\end{equation}
Let $X_{\obs{A}}^{||}$ and $X_{\obs{B}}^{||}$ be the horizontal components of $X_{\obs{A}}$ and $X_{\obs{B}}$.
The Cauchy-Schwarz inequality yields
\begin{equation}\label{estimat}
\begin{split}
\{A,A\}_g\{B,B\}_g
&=G\big(X_{\obs{A}}^{||},X_{\obs{A}}^{||}\big) G\big(X_{\obs{B}}^{||},X_{\obs{B}}^{||}\big)\\
&\geq G\big(X_{\obs{A}}^{||},X_{\obs{B}}^{||}\big)^2+\Omega\big(X_{\obs{A}}^{||},X_{\obs{B}}^{||}\big)^2\\
&= \{A,B\}_g^2+\{A,B\}_\omega^2.
\end{split}
\end{equation}
Estimate \eqref{estimat} together with \eqref{main ett} implies the uncertainty relation \eqref{uncertainty}.
\end{proof}

\subsection{Comparison of the geometric and Robertson-Schr\"{o}dinger uncertainty relations}
We can give a short geometric proof of the Robertson-Schr\"{o}dinger uncertainty relation using the expressions \eqref{geomexp}, \eqref{geoprod},  and \eqref{covariance} for the expectation value functions of $\obs{A}$ and $\obs{B}$
and their symmetric and antisymmetric products: 
\begin{equation}
\begin{split}
\Delta A^2\Delta B^2=
\frac{\hbar^2}{4}\Big(\{A,A\}_g&\{B,B\}_g
+\{A,A\}_g \xi_B^\bot\cdot\xi_B^\bot\\
&+\{B,B\}_g \xi_A^\bot\cdot\xi_A^\bot
+ \left(\xi_A^\bot\cdot\xi_A^\bot\right) \left(\xi_B^\bot\cdot\xi_B^\bot\right)\Big),\label{ettan}
\end{split}
\end{equation}
and
\begin{equation}
\begin{split}
\big((A,B)-AB\big)^2+[A,B]^2
=\frac{\hbar^2}{4}&\Big(\{A,B\}_g^2+\{A,B\}_\omega^2\\
&+2\{A,B\}_g \xi_A^\bot\cdot\xi_B^\bot
+\left(\xi_A^\bot\cdot\xi_B^\bot\right)^2\Big).\label{tvaan}
\end{split}
\end{equation}
Now \eqref{Sur} follows from \eqref{estimat}, \eqref{ettan}, \eqref{tvaan}, and the two estimates
\begin{align}
&\{A,A\}_g \xi_B^\bot\cdot\xi_B^\bot
+\{B,B\}_g \xi_A^\bot\cdot\xi_A^\bot\geq 
2\{A,B\}_g \xi_A^\bot\cdot\xi_B^\bot,\\
&\!\left(\xi_A^\bot\cdot\xi_A^\bot\right) \left(\xi_B^\bot\cdot\xi_B^\bot\right)\geq
\left(\xi_A^\bot\cdot\xi_B^\bot\right)^2.
\end{align}

Apparently, the difference between the geometric and Robertson-Schr\"{o}dinger uncertainty relations lies in the term $2\{A,B\}_g 
\xi_A^\bot\cdot\xi_B^\bot+\left(\xi_A^\bot\cdot\xi_B^\bot\right)^2$. Equation \eqref{tvaan} shows that \eqref{uncertainty} is a geometric equivalent of 
\eqref{Sur} if $2\{A,B\}_g \xi_A^\bot\cdot\xi_B^\bot+\left(\xi_A^\bot\cdot\xi_B^\bot\right)^2=0$.
This is, e.g., the case if $\obs{A}$ or $\obs{B}$ is parallel, or if the density operators in $\D(\sigma)$ represent pure states, in which case 
the vertical bundle has $1$-dimensional fibers. In general, however, the two relations are not equivalent. Equation \eqref{tvaan} also shows that the right hand side of \eqref{uncertainty} interpolates between the two sides of \eqref{Sur} if 
$2\{A,B\}_g \xi_A^\bot\cdot\xi_B^\bot+\left(\xi_A^\bot\cdot\xi_B^\bot\right)^2<0$, and that the 
right hand side of \eqref{Sur} interpolates between the two sides of \eqref{uncertainty} if 
$2\{A,B\}_g \xi_A^\bot\cdot\xi_B^\bot+\left(\xi_A^\bot\cdot\xi_B^\bot\right)^2>0$.
In fact, the maximum of the geometric relation and the Robertson-Schr\"{o}dinger relation 
is a geometric uncertainty relation since the fields $\xi_A$ and $\xi_B$ are intrinsic:
\begin{equation}
\Delta A\Delta B\geq\frac{\hbar}{2}\sqrt{\{A,B\}_g^2+\{A,B\}_\omega^2+\max\{0,2\{A,B\}_g\xi_A^\bot\cdot\xi_B^\bot+(\xi_A^\bot\cdot\xi_B^\bot)^2\}}.
\end{equation}

\subsection{Ensembles of spins}
In this section we give examples of observables for which the geometric uncertainty relation provides a greater 
lower bound for the product of the observables' uncertainties than the Robertson-Schr\"{o}dinger 
relation, and vice versa. But we start with a well-known example for which of the 
two relations provide the same lower bound.

Let $\obs{S}=(\obs{S}_x,\obs{S}_y,\obs{S}_z)$ be the spin operator, and write $\ket{s,m}$ for the
state which is certain to have spin $s$ and magnetic quantum number $m$. Recall that $\ket{s,m}$ is the common eigenstate of
$\obs{S}^2$ and $\obs{S}_z$ such that $\obs{S}^2\ket{s,m}=\hbar^2 s(s+1)\ket{s,m}$ and 
$S_z\ket{s,m}=\hbar m\ket{s,m}$. Also recall that 
$\obs{S}_x=\frac{1}{2}(\obs{S}_+ +\obs{S}_-)$ and $\obs{S}_y=\frac{1}{2i}(\obs{S}_+ -\obs{S}_-)$,
where the raising and lowering operators $\obs{S}_+$ and $\obs{S}_-$ are defined by
\begin{equation}\label{plusminus}
\obs{S}_{\pm}\ket{s,m}=\hbar a^\pm_m\ket{s,m\pm 1},\qquad a^\pm_m=\sqrt{s(s+1)-m(m\pm 1)}.
\end{equation}
We consider an ensemble of spin-$s$ particles prepared so that the proportion of particles having quantum number 
$m_j$ is $p_j$, for $j=1,2,\dots k$. The $m_j$:s should not be confused with the multiplicities of the $p_j$:s. In fact, here we assume that 
that we are in the generic situation where each $p_j$ is nondegenerate.
The spin part of the ensemble's wave function can be 
represented by the density operator $\rho=\sum_j p_j\ketbra{s,m_j}{s,m_j}$. We fix  
$\psi=\sum_j\sqrt{p_j}\ket{s,m_j}\bra{j}$ in the fiber over $\rho$.

The components at $\psi$ of the vector fields on $\S(\sigma)$ associated with $\obs{S}_x$ and $\obs{S}_y$ are  
\begin{align}
X_{\obs{S}_x}(\psi)&=\frac{1}{2i}\sum_{j=1}^k\sqrt{p_j}\big(a^+_{m_j}\ket{s,m_j+1}+a^-_{m_j}\ket{s,m_j-1}\big)\bra{j},\\
X_{\obs{S}_y}(\psi)&=-\frac{1}{2}\sum_{j=1}^k\sqrt{p_j}\big(a^+_{m_j}\ket{s,m_j+1}-a^-_{m_j}\ket{s,m_j-1}\big)\bra{j}.
\end{align}
These vectors are horizontal because the spectrum of $\rho$ is nondegenerate and the matrices $[\bra{i}\psi^\dagger X_{\obs{S}_x}(\psi)\ket{j}]$ 
and $[\bra{i}\psi^\dagger X_{\obs{S}_y}(\psi)\ket{j}]$ have only zeros on their diagonals. Moreover, they are orthogonal since $(S_x,S_y)=0$. 
However, $[S_x,S_y]\ne 0$:
\begin{equation}
\{S_x,S_y\}_\omega(\rho)=\hbar\sum_{j=1}^k p_jm_j.
\end{equation}
We conclude that
\begin{equation}\label{sista}
\Delta S_x(\rho)\Delta S_y(\rho)\geq\frac{\hbar^2}{2}\Big|\sum_{j=1}^k p_jm_j\Big|.
\end{equation}
One identifies the right hand side of \eqref{sista} as 
$\hbar/2$ times the modulus of $S_z(\rho)$.

Next define four observables 
\begin{equation}
\obs{A}=\obs{S}_x+\sqrt{\eps}\obs{S}_z,
\quad\obs{B}=\obs{S}_x-\sqrt{\eps}\obs{S}_z,
\quad\obs{C}=\obs{S}_x+\obs{S}_z,
\quad\obs{D}=\obs{S}_y+\obs{S}_z,
\end{equation}
where $0<\eps<1/2s$. We have that
\begin{align}
&X_{\obs{A}}^{||}(\psi)=X_{\obs{B}}^{||}(\psi)=X_{\obs{C}}^{||}(\psi)=X_{\obs{S}_x}(\psi),\quad X_{\obs{D}}^{||}(\psi)=X_{\obs{S}_y}(\psi),\\
&\xi_A(\rho)=\sqrt{\eps}\xi_{S_z}(\rho),\quad \xi_B(\rho)=-\sqrt{\eps}\xi_{S_z}(\rho),\quad\xi_C(\rho)=\xi_D(\rho)=\xi_{S_z}(\rho),
\end{align}
because $\obs{S}_x$ and $\obs{S}_y$ are parallel and $\obs{S}_z$ is perpendicular at $\rho$. Consequently,
\begin{align}
&\{A,B\}_g(\rho)+\xi_A^\bot(\rho)\cdot\xi_B^\bot(\rho)
=\{S_x,S_x\}_g(\rho)-\eps\xi_{S_z}^\bot(\rho)\cdot\xi_{S_z}^\bot(\rho),\\
&\{C,D\}_g(\rho)+\xi_C^\bot(\rho)\cdot\xi_D^\bot(\rho)=\{S_x,S_y\}_g(\rho)+\xi_{S_z}^\bot(\rho)\cdot\xi_{S_z}^\bot(\rho).
\end{align}

As mentioned before, the vectors $X_{\obs{S}_x}(\psi)$ and $X_{\obs{S}_y}(\psi)$ are orthogonal. Moreover, $\xi_{S_z}^\bot(\rho)$ is nonzero. 
Therefore
$0=\{C,D\}_g(\rho)<\{C,D\}_g(\rho)+\xi_C^\bot(\rho)\cdot\xi_D^\bot(\rho)$.
From this it follows that the the Robertson-Schr\"{o}dinger relation provides a greater lower bound for $\Delta C\Delta D$ at $\rho$ than the
geometric uncertainty relation. However, the situation is reversed for the observables $\obs{A}$ and $\obs{B}$. Explicitly,
\begin{equation}
\!\{S_x,S_x\}_g(\rho)=\hbar s(s+1)-\hbar\sum_{j=1}^k m_j^2 p_j,
\end{equation}
and 
\begin{equation}
\xi_{S_z}^\bot(\rho)\cdot\xi_{S_z}^\bot(\rho)=2\hbar\sum_{j=1}^k m_j^2p_j-2\hbar\Big(\sum_{j=1}^k m_jp_j\Big)^2.
\end{equation}
The bounds on $\eps$ are chosen such that
\begin{equation}
0<2\eps\hbar\sum_{j=1}^k m_j^2p_j-2\eps\hbar\Big(\sum_{j=1}^k m_jp_j\Big)^2<\hbar s(s+1)-\hbar\sum_{j=1}^k m_j^2 p_j.
\end{equation}
Accordingly, $0<\{A,B\}_g(\rho)+\xi_A^\bot(\rho)\cdot\xi_B^\bot(\rho)<\{A,B\}_g(\rho)$.
From this it follows that the geometric uncertainty relation provides a greater lower bound for $\Delta A(\rho)\Delta B(\rho)$ than
the Robertson-Schr\"{o}dinger relation.

\section{Conclusion}
In this paper we have equipped the orbits of isospectral finite rank density operators
with Riemannian and symplectic structures, and we have derived a geometric uncertainty principle for observables acting on quantum systems in mixed states. Moreover, we have compared this uncertainty relation with the Robertson-Schr\"{o}dinger uncertainty relation.
It turned out that the two relations in general are not equivalent.

\end{document}